\documentclass[12pt]{article}
\usepackage{setspace}
\usepackage{geometry}
\usepackage{color}
\usepackage{amsfonts}
\usepackage{amssymb}
\usepackage{amsmath}
\usepackage[mathscr]{eucal}
\usepackage{amsfonts}
\usepackage{amsmath,amsxtra,latexsym,amsthm, amssymb, amscd}
\usepackage[colorlinks,citecolor=blue,filecolor=black,linkcolor=blue,urlcolor=black]{hyperref}
\usepackage{pgfplots} 
\usepackage{url}
\usepackage[round]{natbib}
\usepackage{fancyhdr}
\usepackage{tikz}

\usepackage[round]{natbib}
\usepackage{fancyhdr}
\usepackage{makeidx}

\makeindex

\usepackage{graphicx}
\usepackage{hyperref}
\usepackage{subcaption}
\usepackage{caption}
\usepackage{epstopdf}

\usepackage{soul}

\newtheorem{theorem}{Theorem}

\newtheorem{definition}{Definition}

\newtheorem{lemma}{Lemma}

\newtheorem{proposition}{Proposition}
\newtheorem{remark}{Remark}

\newtheorem{assum}{Assumption}

\newcommand{\rr}{\mathbb{R}}

\newcommand{\summ}{\sum\limits}

\newcommand{\ma}{\max\limits}

\newcommand{\liminff}{\liminf\limits}
\newcommand{\limsupp}{\limsup\limits}
\newcommand{\limm}{\lim\limits}

\begin{document}
\title{FDI versus R\&D in an endogenous growth model\footnote{The authors thank anonymous referees for their insightful comments and suggestions.}}

\author{Thanh Tam NGUYEN-HUU\thanks{ORCID ID: \href{https://orcid.org/my-orcid}{0000-0001-6953-9838}. EM Normandie Business School, M\'etis Lab (France). Email: tnguyenhuu@em-normandie.fr. Phone: +33 (0)2 78 34 04 61. Address: EM Normandie (campus  Havre),  20 Quai Frissard, 76600 Le Havre, France.} 
\and  Ngoc-Sang PHAM\thanks{Corresponding author. ORCID ID: \href{https://orcid.org/0000-0001-9037-3019}{0000-0001-9037-3019}. EM Normandie Business School, M\'etis Lab (France). Email: npham@em-normandie.fr.  Phone:  +33 (0)2 50 32 04 08. Address: EM Normandie (campus Caen), 9 Rue Claude Bloch, 14000 Caen, France.}}
\date{\today}

\maketitle

\begin{abstract}
We investigate the role of foreign direct investment (FDI) and research and development (R\&D) in the transitional dynamics of host countries using an optimal growth model. FDI may benefit the host country's GNP by enabling multinational enterprises to hire local workers. However, if the host country focuses solely on FDI, it may fall into a middle-income trap. Most importantly, we show that if the host country invests in R\&D, its economy can reach sustained growth. In this case, FDI benefits the host country, but only in the early stages of its development process. 
\end{abstract}

\section{Introduction}
\label{introduction_section}
Over the past few decades, openness to the global economy and the attraction of foreign direct investment (FDI) have become major policy priorities for developing countries seeking to foster economic development. The prevailing argument is that multinational enterprises (MNEs) stimulate investment, bring advanced technologies and managerial expertise, and generate positive spillovers for domestic firms. Despite these expectations, the impact of FDI on host-country development remains ambiguous.

Overall, the empirical literature finds that FDI's effect on host-country economic growth is relatively weak \citep{cl2005, g17}, and the relationship between FDI and growth varies over time \citep{benetrix23}. More precisely, the significance of this effect depends on local conditions such as income levels, institutional quality \citep{bg20}, human capital accumulation \citep{li05}, and the development of local financial markets \citep{alfaro04, alfaro10}. For example, the impact of FDI on growth tends to be positive in countries with relatively high levels of human capital or well-developed financial systems \citep{borenz98, alfaro04}. Moreover, the relationship between FDI and income levels follows an inverted U-shape: the effect is strongest in low- to middle-income countries and diminishes as countries transition toward high-income status \citep{bg20}.\footnote{See, for example, \cite{amin14} for a comprehensive review of the literature on the FDI–growth relationship.}

Despite the large body of empirical research on the FDI–growth nexus,\footnote{Over five decades of research on FDI, the FDI–GDP (economic growth) relationship has been the most extensively studied. Indeed, 107 out of the 500 articles reviewed in \cite{pf21} analyze the impact of FDI on host-country economic growth.} theoretical analyses remain relatively scarce. This paper aims to fill this gap by examining the role of FDI along the transitional dynamics of a host economy. In particular, we address the following fundamental questions:
\begin{itemize}
\item[(i)] What is the optimal strategy for a country receiving FDI?
\item[(ii)] What role does FDI play in the host country’s development process? Can it help the host economy escape the middle-income trap and achieve sustained long-term growth?
\end{itemize}

To answer these questions, we develop an optimal growth model with FDI and endogenous growth. The host country is modeled as a small open economy producing three goods: consumption, physical capital, and new goods. All of which are freely tradable internationally. The economy consists of two agents: a representative domestic agent and an MNE. The representative agent faces three investment strategies: (i) investing in physical capital to produce consumption goods; (ii) investing in training to acquire skills and work for the MNE in exchange for wages; and (iii) investing in research and development (R\&D) to generate innovations. Successful innovations enhance the productivity of domestic firms.

First, we show, in Theorem \ref{poverty2},  that when the host country has low initial resources and weak research efficiency, it is optimal to forgo R\&D investment and focus exclusively on attracting FDI. In this case, the economy converges to a higher steady state than that without FDI.

Second, we consider a low-income country that cannot immediately invest in R\&D or new technologies due to high fixed costs. In this context, we demonstrate that if the productivity gains from new technologies are sufficiently large or if the country possesses strong R\&D potential, the optimal development strategy follows a three-stage process (see Theorems \ref{dynamic_newtech} and \ref{theorem_nofixedcots}):
\begin{itemize}
\item \textbf{Stage 1:} the country invests in training specialized workers;
\item \textbf{Stage 2:} these workers are employed by MNEs, earning higher wages and contributing to income growth and capital accumulation;
\item \textbf{Stage 3:} once income reaches a sufficient level, the country shifts toward R\&D investment, generating innovations that raise domestic firms’ total factor productivity (TFP) and enable sustained long-run growth.
\end{itemize}

Our model also shows that the country may achieve long-run growth without FDI. In this sense, FDI acts primarily as a catalyst, particularly during the early stages of development, rather than as a fundamental driver of long-term growth.

This paper makes two significant contributions to the literature. First, it advances the theoretical understanding of the mechanisms through which FDI affects economic growth. Early contributions include \cite{Findlay78}, who analyzes FDI in a dynamic framework with exogenously determined technological efficiencies and introduces a “contagion” effect whereby domestic firms’ efficiency depends on that of foreign firms.\footnote{In this framework, domestic and foreign firms’ capital stocks evolve according to a continuous-time dynamical system with exogenous parameters.} Building on this idea, \cite{WANG1990} introduces technology diffusion by modeling the host country’s human capital stock as an increasing function of the ratio of foreign to domestic capital. In a two-country model with free capital mobility and exogenous saving rates, \cite{WANG1990} shows that openness to FDI benefits the host economy.

Subsequent studies relax the assumption of exogenous saving behavior. In a continuous-time endogenous growth model with a continuum of capital-good varieties,\footnote{See, for example, \cite{romer90} and \cite{gh91}.} \cite{borenz98} model FDI as the share of product varieties produced by foreign firms. Under Cobb–Douglas production and CRRA utility, they show that the steady-state growth rate increases with the foreign share of varieties. \cite{berth00} extends this framework by endogenizing the number of varieties produced by domestic and foreign firms. Similarly, \cite{alfaro10} examine how financial market development mediates the growth effects of FDI through backward linkages, focusing on balanced growth paths and showing that FDI generates more substantial growth effects in financially developed economies.

Unlike these studies, we do not restrict attention to steady states or balanced growth paths, nor do we impose specific functional forms on preferences. Instead, we analyze the global and transitional dynamics of optimal growth paths in models with and without FDI.\footnote{\cite{np18,NguyenPham2024} examine FDI spillovers and industrial policy in two-period and exogenous growth models, respectively, but do not consider endogenous growth.} To the best of our knowledge, our paper is the first to adopt such an approach. As a result, our findings are more robust and provide deeper insights into optimal development strategies over time, an aspect largely absent from the existing literature.

Second, our paper contributes to the literature on optimal growth with thresholds \citep{azariadis1990threshold, levan09, levan10, lv16} and increasing returns \citep{romer86, jm90, kamiroy07}. Our contribution lies in highlighting the role of FDI in helping economies overcome early-stage development thresholds by supporting capital accumulation. However, we show that long-run growth ultimately depends not on FDI itself, but on domestic conditions, particularly innovation capacity and R\&D efficiency. From a technical perspective, our analysis is far from trivial due to the coexistence of domestic and foreign firms, which prevents the direct application of standard methods such as those used in \cite{levan09} and \cite{levan10}. Moreover, in the presence of thresholds and possibly increasing returns to scale, the payoff function is non-smooth and non-convex.

The remainder of the paper is organized as follows. Section \ref{structure} presents the endogenous growth model with FDI. Section \ref{mainresult} analyzes the interaction between FDI, R\&D, and host-country growth. Section \ref{conclusion} concludes. Formal proofs are relegated to  Appendix \ref{appendix-proof}.

\section{A growth model with FDI and R\&D}\label{structure}

Let us consider a small open economy producing three goods: a consumption good, physical capital, and a so-called \textit{new good}. The consumption good is chosen as the numéraire. The price of physical capital, expressed in units of the consumption good, is exogenous and denoted by $p$.

In each period, there is a representative MNE in the host country. The MNE produces the new good using two inputs: physical capital and specific labor. We assume that no domestic firms operate in this sector.

At each date $t$, the foreign firm (without market power) chooses the quantities of physical capital $K_{e,t}$ and specific labor $L^D_{e,t}$ to maximize its profit:
\begin{eqnarray} \label{e_MNE}(F_t): \quad \pi_{e,t}=&&\ma_{K_{e,t},L^D_{e,t}\geq 0} \Big[p_nF^e_t(K_{e,t},L^D_{e,t})-pK_{e,t}-w_tL^D_{e,t}\Big]
\end{eqnarray}
where $p_n$ is the exogenous price (in terms of consumption good) of the new good, and $w_t$ is the endogenous wage rate. The production function is defined by $F^e_t(K,L)= A_eK^{\alpha_e}L^{1-\alpha_e}$ $\forall (K,L)\in \rr_+^2$, where $\alpha_e\in (0,1)$ and $A_e>0$,

The host country is populated by a representative agent who treats prices and wages as given and chooses allocations to maximize the population's intertemporal welfare. At each date $t$, this agent has three investment options.  First, he(she) can invest $K_{c,t+1}$ units of physical capital to produce $A_c K_{c,t+1}^{\alpha}$ units of the consumption good in period $t+1$, where $\alpha\in(0,1)$. Second, he(she) invest $H_{t+1}$ units of the consumption good in training to generate $A_h H_{t+1}^{\alpha_h}$ units of specific labor, where $\alpha_h\in(0,1)$. This labor is supplied to the MNE, yielding a total wage income of $w_{t+1} A_h H_{t+1}^{\alpha_h}$. Third, he(she) can also invest $N_{t+1}$ units of the consumption good in R\&D at period $t$ to create new technology. This generates $b N_{t+1}^{\sigma}$ units of new technology in period $t+1$, where $b>0$ captures the efficiency of the research process and $\sigma\in (0,1]$. New technology improves productivity in the consumption sector only if investment exceeds a fixed threshold, that is, if $bN_{t+1}^{\sigma}>\bar{x}$, where $\bar{x}\geq 0$. In this case, the productivity becomes $$A_c+a(bN_{t+1}^{\sigma}-\bar{x})$$ where $a>0$ measures the effectiveness of new technology.\footnote{R\&D could be introduced in other ways, for instance, $A_c+\gamma((N_{t+1}-N^*)^+)^{\sigma}$. However, our main results remain qualitatively unchanged.}

The representative agent maximizes the intertemporal utility $\summ_{t=0}^{+\infty} \beta^t u(c_{t})$ 
subject to
\begin{subequations}\label{budgetconstraints}
\begin{align}
 c_{t}+pK_{c,t+1}+N_{t+1}+H_{t+1}&\leq \Big(A_c+a(bN_t^{\sigma}-\bar{x})^+\Big)K_{c,t}^{\alpha}+w_tL_{e,t}\\
 L_{e,t}&\leq A_hH_t^{\alpha_h},
\end{align}
\end{subequations}
and $c_{t},K_{c,t}, H_t, L_{e,t}, N_t\geq 0$ $\forall t$. Here, $\beta\in (0,1)$ is a rate of time preference, while $u$ is the instantaneous utility function.

We now introduce the notion of equilibrium.

\begin{definition}\label{def} An intertemporal equilibrium is a sequence $ (c_t,K_t,H_t, N_t, L_{e,t}, L^D_{e,t}, K_{e,t}^D, w_t)_{t=0}^{\infty}$ satisfying three conditions: (i) Given $(w_t)_{t=0}^{\infty}$, the allocation $(c_t,K_t,H_t, N_t, L_{e,t})_{t=0}^{\infty}$ solves the representative agent's problem, (ii) Given $w_t$, the allocation $( L^D_{e,t}, K_{e,t}^D)$ solves problem $(F_t)$, (iii) The labor market clears: $L^D_{e,t}=L_{e,t}$.
\end{definition}

We impose in our paper standard assumptions.
\begin{assum}\label{assumption_basic}
The utility function $u$ is in $C^1$, strictly increasing, concave, and $u'(0)=\infty$.  Technology parameters satisfy $A_c>0, A_h>0, \alpha \in (0,1), \alpha_h\in (0,1), \sigma\in (0,1]$. Initial conditions: $N_{0}=0$ while $K_{c,0}, L_{e,0}>0$. 
 \end{assum}
 Without an explicit statement to the contrary, we also require the following assumption.
\begin{assum}\label{assumption_rd}
$\sigma\in (0,1)$ and $a\bar{x}>A_c$ (this means that the fixed cost $\bar{x}$ is not too low). 
 \end{assum}
 
We first derive the equilibrium wage.
\begin{lemma}\label{wage}Under Assumption \ref{assumption_basic}, in equilibrium we have
\begin{align}\label{wt}
w_t=w\equiv\Big(\alpha_e^{\alpha_e}(1-\alpha_e)^{1-\alpha_e}\frac{p_nA_e}{p^{\alpha_e}}\Big)^{\frac{1}{1-\alpha_e}} \text{ } \forall t.
\end{align} 
\end{lemma}
\begin{proof}See Appendix \ref{appendix-proof}.\end{proof}
Note that the equilibrium wage depends not only on the MNE's productivity $A_e$ but also on the prices of physical capital and of new goods, $p$ and $p_n$.

Define the aggregate savings (and investment) of the country as
$$S_{t+1}=pK_{c,t+1}+N_{t+1}+H_{t+1}.$$ 
In our framework, it equals the aggregate investment. Using Definition \ref{def} and Lemma \ref{wage}, we obtain the relationship between the equilibrium path and the solution of an optimal growth model.
\begin{lemma}[equilibrium and optimal growth]Under Assumption \ref{assumption_basic}, the equilibrium investment $S_t$ is part of the solution to the following optimal growth problem.
\begin{eqnarray} (P'): &\ma_{(c_{t},S_{t+1})_{t=0}^{+\infty}} \Big[\summ_{t=0}^{+\infty} \beta^t u(c_{t})\Big] \text{ subject to: }  c_t, S_t\geq 0, \text{ } 
 c_{t}+S_{t+1}\leq G(S_t) \text{ } \forall t\geq 1,
\end{eqnarray}
 and $c_0+S_1\leq X_0,$ where $X_0\equiv A_cK_{c,0}^{\alpha}+w_0L_{e,0}$ and $G(S)$ is defined by 
\begin{subequations}\label{problemGS}
\begin{align}
\label{n1}(G_S):  G(S)&\equiv \ma_{K_{c}, N, H}\Big\{g(K_c,N,H): pK_c+N+H\leq S; K_{c}, N, H\geq 0 \Big\},\\
\text{where }\label{n2}& g(K_c,N,H)\equiv \Big(A_c+a(bN^{\sigma}-\bar{x})^+\Big)K_{c}^{\alpha}+wA_hH^{\alpha_h}.
\end{align}
\end{subequations}
\end{lemma}
According to this result, we will focus on the optimal growth problem $(P')$. Notice that the function $G(\cdot): \rr_+\to \rr_+$ is continuous, strictly increasing and 
$G(0)=0$. However, it may be neither concave nor smooth.

\begin{remark}[No FDI and no R\&D] In the absence of both the MNE and R\&D, we recover an economy without FDI. In this case, the representative agent's problem reduces to the standard Ramsey optimal growth model with the budget constraint: $c_{t}+pK_{c,t+1}\leq A_cK_{c,t}^{\alpha}$ $\forall t$. In this case, $\lim_{t\rightarrow \infty}S_t={S}_a$, where ${S}_a$ is defined by ${S}_a^{1-\alpha}=\alpha\beta A_c/p^{\alpha}$.
\end{remark}

Consider now the case with FDI but without R\&D.  We define the function $F:\rr_+\to\rr_+$ by $
F(S)\equiv \ma_{pK_{c}+H\leq S, K_{c}\geq 0, H\geq 0}\{A_cK_{c}^{\alpha}+w A_hH^{\alpha_h} \}$ $\forall S\geq 0$. The representative agent's problem becomes 
\begin{align}
 (P_1'):  &\ma_{(c_{t},S_{t+1})_{t=0}^{+\infty}} \Big[\summ_{t=0}^{+\infty} \beta^t u(c_{t})\Big] 
\text{ subject to }c_t,S_t\geq 0,   c_{t}+S_{t+1}\leq F(S_t).
\end{align}

Since $ \alpha,\alpha_h\in (0,1)$, we can prove the following result.
\begin{lemma}\label{function_F}
The function $F$ is strictly increasing, strictly concave, continuously differentiable and satisfies the Inada condition $\lim_{x\to 0}F'(x)=\infty$. 
\end{lemma}

 Using the standard argument in dynamic programming, we obtain the following result. 
\begin{proposition}[with FDI but without R\&D]\label{noRD} Assume that the MNE is present, but no investment in R\&D occurs in the host country. Then $S_t$ converges to ${S}_b$, where $S_b$ is uniquely defined by 
\begin{align}\label{Sb}
\beta F'(S_b)=1.
\end{align}  Moreover, $S_b$ increases in $A_c, w,A_h$, and  $S_b>S_a$.   
 
In particular, if $\alpha=\alpha_h$, then we can obtain an explicit form of $S_b$ as
\begin{align}\label{Sba}
{S}_b^{1-\alpha}=\alpha\beta A \text{ where }
A\equiv \big((\dfrac{A_c}{p^{\alpha}})^{\frac{1}{1-\alpha}} +(w A_h)^{\frac{1}{1-\alpha}}  \big)^{1-\alpha}.
\end{align}
    
\end{proposition}
 
The inequality $S_b>S_a$ implies that the presence of the MNE raises the steady-state level of investment relative to an economy without FDI. Moreover, the steady-state level $S_b$ increases with the productivity of the traditional sector, wages, and the MNE’s productivity. This highlights that the growth effects of FDI depend jointly on foreign investment and host-country characteristics, a result consistent with the empirical findings discussed in the Introduction.

\section{Global analysis: role of FDI and R\&D}\label{mainresult}
We now investigate the global dynamics of the allocation to explore the role of FDI. We first provide  static analysis in Subsection \ref{ss_static} and then global dynamic analysis  in Subsection \ref{ss_dynamic}.

\subsection{Static analysis}\label{ss_static}

In this subsection, given the savings $S>0$, we study the optimal allocation of the host country. Formally, we look at the optimization problem $(G_S)$. First, it is easy to see that this problem has a solution.\footnote{However, since the objective function is not concave, the uniqueness of solutions may not be ensured.} Then, we have the following result. 
\begin{proposition}\label{s2} Let Assumption \ref{assumption_basic} be satisfied. 
Consider the optimization problem $(G_S)$ described in \eqref{problemGS}.
\begin{enumerate}
	\item \label{static_1} If $bS^{\sigma}\leq \bar{x}$, then at optimum, we have $N=0$  $\forall a$.
	\item  \label{static_2} If $bS^{\sigma}> \bar{x}$ and $\Big[A_c+{a}\Big(\big(b^{\frac{1}{\sigma}}\frac{S}{2}+\frac{\bar{x}^{\frac{1}{\sigma}}}{2}\big)^{\sigma}-\bar{x}\Big)\Big]\frac{1}{p^{\alpha}}  \Big(\frac{S}{2}-\frac{\bar{x}^{\frac{1}{\sigma}}}{2b^{\frac{1}{\sigma}}}\Big)^{\alpha}>F(S)$, then $N>0$ at optimum.
\item \label{static_3}Let $\lambda$ be higher than 
$\max\{1,2^{\sigma}\}$ and define   $M\equiv \frac{1}{2p^{\alpha}}\big(\frac{1}{2^\sigma}-\frac{1}{\lambda}\big)\big(1-\frac{1}{\lambda^{1/\sigma}}\big)^{\alpha}>0$. Assume also that $\alpha+\sigma>\alpha_h$.

If $S>\bar{S}(a,b)\equiv \max\{\big(\frac{\lambda \bar{x}}{b}\big)^{\frac{1}{\sigma}}, \big(\frac{2A_c}{p_c^{\alpha}abM}\big)^{\frac{1}{\sigma}},\big(\frac{2wA_h}{abM}\big)^{\frac{1}{\sigma+\alpha-\alpha_h}}\}$, then $N>0$ in optimal.
\end{enumerate}
\end{proposition}

\begin{proof} 
See Appendix \ref{appendix-proof}.
\end{proof}
The intuition behind point \ref{static_1} of Proposition \ref{s2} is that the host country may choose not to invest in R\&D if either the efficiency of the research process or the initial resource level is low, or if the fixed cost is high. In addition, points \ref{static_2} and \ref{static_3} indicate that the country invests in R\&D when $a$ and $b$ are sufficiently large (since $F(S)$ depends on neither $a$ nor $b$).

\subsection{Global dynamic analysis}\label{ss_dynamic}
In this section, we explore the global dynamics of the host economy and showing the role of FDI.  First, we have the monotonicity of the savings path $(S_t)$.

\begin{proposition}
\label{basic}Let Assumption \ref{assumption_basic} be satisfied.  The optimal path $(S_t)_t$ is monotonic. Moreover, $S_t$ does not converge to zero.
    
\end{proposition}
\begin{proof}See Appendix \ref{appendix-proof}.\end{proof}

Second, we study the boundedness of the allocation. Let us define the sequence $(x_t)$ as $x_0=X_0, x_{t+1}=F(x_t)$. Let $x^*$ and $\bar{S}$ be uniquely defined by:
\begin{align}\label{sbar}
F(x^*)=x^* \text{ and }
\bar{S}\equiv\ma \{X_0, x^*\}.
\end{align}

Notice that $x^*$ and $\bar{S}$ depend on (i) the productivity $A_c$ and capital elasticity $\alpha$ of the consumption good sector, (ii) the efficiency of specific labor training $A_h, \alpha_h$, and (iii) wage $w$.\footnote{If $\alpha_h=\alpha$, we can explicitly compute that $x^*=({A_c}/{p^{\alpha}})^{\frac{1}{1-\alpha}} +(w A_h)^{\frac{1}{1-\alpha}}.$}

It is important to mention some properties of the function $F$ and the threshold $\bar{S}$.
\begin{lemma}\label{FS}(1) $F(x)\leq F(x^*)= x^*$ $\forall x\leq x^*$ and $F(x)\leq x$ $\forall x\geq x^*$. (2) $x_t<\bar{S}$ $\forall t$
\end{lemma}
\begin{proof}See Appendix \ref{appendix-proof}.\end{proof}
We are ready to state the following result.
 
\begin{theorem}[fixed cost and middle income trap] \label{poverty2} Let Assumption \ref{assumption_basic} be satisfied.  Assume that $X_0\equiv A_cK_{c,0}^{\alpha}+w_0L_{e,0}\leq x^*$ and $b(x^*)^{\sigma}\leq \bar{x}$, where $x^*$ is defined by (\ref{sbar}). Then, in optimal, $N_t=0$ $\forall t$. In this case, we have $\lim_{t\rightarrow \infty}S_t=S_b$, which is the steady state investment in the economy with FDI but without R\&D ($S_b$ is defined in Proposition  \ref{noRD}).
\end{theorem}
\begin{proof}See Appendix \ref{appendix-proof}.\end{proof}

Theorem \ref{poverty2} indicates that when the host country has  a low initial resource (in the sense that $X_0\equiv A_cK_{c,0}^{\alpha}+w_0L_{e,0}\leq x^*$) and a high fixed cost (in the sense that $b(x^*)^{\sigma}\leq \bar{x}$), it never invests in R\&D ($N_t=0$ for $\forall t$).\footnote{In Proposition \ref{theorem_nofixedcots_1} in Appendix \ref{proofs_nofixedcost}, we show that $N_t$ may be zero even when there is no fixed cost.} In this case, both savings $S_t$ and output are bounded from above. We refer to this situation as a middle-income trap. 

In Theorem \ref{poverty2}, it is interesting to note that due to the high fixed cost $(\bar{x})$, a middle-income trap arises regardless the level of effectiveness $(a)$ of new technology. The next result shows that a middle-income trap can also arise without any fixed cost: if productivity and R\&D efficiency are sufficiently low, the economy will not experience sustained long-run growth.

\begin{proposition}[efficiency and middle-income trap]
\label{bounded_nofixedcost}
Let Assumption \ref{assumption_basic} be satisfied. Let $\bar{X}>X_0$. Assume that productivity and efficiency (parameters $A_c,A_h$ and $a,b$) are low in the sense that
\begin{align}\label{assumption_lowproductivity}
3wA_h&<\bar{X}^{1-\alpha_h},& 3A_c&<p^{\alpha}\bar{X}^{1-\alpha}, & 3ab\bar{X}^{\alpha+\sigma-1}<p^{\alpha}.
\end{align}
Then, $S_t\leq \bar{X}$ $\forall t\geq 1$.
\end{proposition}
\begin{proof}See Appendix \ref{appendix-proof}.\end{proof}


We now explore conditions under which the economy may grow without bound.
\begin{theorem}[convergence and growth with increasing return to scale] \label{dynamic_newtech}Let Assumptions \ref{assumption_basic} and \ref{assumption_rd} be satisfied. 
Assume also that $\alpha +\sigma \geq 1$.\footnote{Some empirical research likely supports the condition of increasing returns to scale, for example, \cite{shank24} relies on the methods of \cite{olley96} and \cite{ackerberg15} to estimate the aggregate production function of US industries, and the author observes that the US economy exhibits increasing returns to scale over the period 1980-2019. Meanwhile, \cite{romeo17} find evidence of the contribution of the research intensity to productivity growth and the existence of increasing returns to scale.} When $a,b$ are high enough, the optimal path $(S_t)$ converges to infinity: $\lim_{t\rightarrow \infty}S_t=\infty$. Moreover, 
\begin{align}\label{converge_share}
\limm_{t\rightarrow \infty}\frac{N_t}{S_t}=\frac{\sigma}{\alpha +\sigma }, \quad  \limm_{t\rightarrow \infty}\frac{pK_{c,t}}{S_t}=\frac{\alpha}{\alpha +\sigma },\quad  \limm_{t\rightarrow \infty}\frac{H_t}{S_t}=0.
\end{align}

\end{theorem}
\begin{proof}See Appendix \ref{appendixA1}.\end{proof}

The sustained growth (that is the case where  $\lim_{t\rightarrow \infty}S_t=\infty$) in Theorem  \ref{dynamic_newtech} is based on two key points: (1) increasing returns to scale (i.e., $\alpha+\sigma\geq 1$) and (2) a high R\&D efficiency ($b$) and a high effectiveness ($a$) of new technology.

Note that the conditions given in Theorem  \ref{dynamic_newtech} do not depend on the initial resource $X_0\equiv A_cK_{c,0}^{\alpha}+w_0L_{e,0}$ which is less than $x^*$.  So, our theoretical results lead to an interesting implication: Consider a low-income country characterized by condition $bX_0^{\alpha}< \bar{x}$. According to Proposition \ref{s2}, we have $N_1=0$, i.e., the country cannot immediately improve the local firm TFP. Now, suppose that the leverage of new technology $a$ is high enough and conditions in Theorem \ref{dynamic_newtech} hold. In this case, the country obtains a sustained growth (in the sense that $\lim_{t\rightarrow \infty}S_t=\infty$). Moreover, the sequence $S_t$ is increasing in time. According to point (ii) of Proposition \ref{s2}, there is a date $t_0$ along the optimal path such that the country should focus on R\&D from date $t_0$ on (i.e., $N_t=0$ $\forall t\leq t_0$ and $N_t>0$  $ \forall t>t_0$). Therefore, the optimal strategy of the country should be as follows. 
%
%
%
\begin{enumerate}
	\item First, the country should train specific workers. 
	\item Second, specific workers will work for the MNE to improve the country's income and capital accumulation. 
	\item Third, once the country's resource is high enough, it should focus on R\&D to create new technology that increases the country's TFP. Hence, its economy may grow faster and converge to a high-income country.
\end{enumerate}



In Theorem  \ref{dynamic_newtech}, we have assumed that the fixed cost $\bar{x}$ is strictly positive. However, the following result shows that the main results, including the insights of the three-stage development path, remain valid for the case without fixed costs.
\begin{theorem}\label{theorem_nofixedcots} Let Assumption \ref{assumption_basic} be satisfied. Assume that there is no fixed cost ($\bar{x}=0$) and $\sigma\in (0,1)$. If $\alpha+\sigma\geq 1$ and 
\begin{align}\label{highproductivity_nofixedcost}
\min\Big\{\big(\frac{\alpha}{p}\big)^{\alpha}(\sigma ab)^{1-\alpha}A_c^{\alpha},\sigma^{1-\alpha} ab\big(\frac{\alpha}{p}\big)^{\alpha} \Big\}>1/\beta,
\end{align}
then the conclusions of Theorem \ref{dynamic_newtech} hold. In particular, $H_1/S_1>\lim_{t\to\infty}H_t/S_t=0$.
\end{theorem}
\begin{proof}See Appendix \ref{proofs_nofixedcost}.\end{proof}
In the early stages of development, the ratio $H_t/S_t$ is significant, reflecting that investment in training firm-specific labor for the multinational enterprise accounts for a substantial share of total investment. As the economy develops, however, this ratio converges to zero in the long run.

Theorems \ref{dynamic_newtech} and \ref{theorem_nofixedcots} are closely related to the literature on economic growth with increasing returns to scale \citep{romer86, jm90, levan09, levan10}.\footnote{Theorem \ref{theorem_nofixedcots} can be viewed as an extension of Proposition 3 in  \cite{LeVanPham2022}, where they also study sustained growth but abstract from FDI.} Our main contribution is to incorporate foreign direct investment (FDI) into an optimal growth framework and to characterize its role along the development path. In our model, FDI benefits the host country primarily during the early stages of development by facilitating resource accumulation. However, the property $\lim_{t\rightarrow \infty}S_t=\infty$ 
together with condition (\ref{converge_share}), indicates that in the long run—once the host country’s resource base becomes sufficiently large—the economy must shift its focus toward domestic investment in physical capital and R\&D in order to sustain growth, rather than continuing to rely on FDI. In this sense, FDI acts as a transitional growth engine rather than a determinant of long-run growth.


\begin{remark}[growth without FDI]It is worth noting that the conditions in Theorems \ref{dynamic_newtech} and \ref{theorem_nofixedcots} may be satisfied even when $A_e=w=0$. In other words, a host country can achieve long-run economic growth even in the absence of FDI. 

 \end{remark}

 It should be noticed that the condition of increasing returns to scale $(\alpha+\sigma\geq 1)$ plays an important role in growth without bounds. The following result proves that it is actually essential: In the case of decreasing returns to scale, the capital stock may converge to a finite steady state.  

\begin{theorem}[decreasing return to scale]\label{drs}Let Assumption \ref{assumption_basic} be satisfied. 
Assume that $\alpha +\sigma<1$. The optimal path $(S_t)$ converges to a finite value. 
\end{theorem}
\begin{proof}See Appendix \ref{proof_drs}.\end{proof}

Drawing on Theorem \ref{poverty2}, Proposition \ref{bounded_nofixedcost} and Theorem \ref{drs}, we can conclude that the middle-income trap (i.e., when $\limsup_{t\to\infty}S_t<\infty$) is mainly driven by three factors: (1) high fixed costs, (2) low efficiency (captured by low values of $a$ and $b$), and (3) decreasing returns to scale.

\paragraph{Discussion.} So far we have presented several theoretical results regarding the role of FDI in the host economy. In general, host countries benefit from inward FDI. However, the impact of FDI on economic growth depends not only on the nature of FDI itself but, more importantly, on the host country’s underlying conditions, such as its initial resource endowment, the productivity of domestic firms, the quality of the education system, and the efficiency of the R\&D process.

In particular, as shown in Proposition \ref{noRD}, if the host country relies solely on FDI, the steady-state level $S_b$, which exceeds that of an economy without FDI—is increasing in local structural factors, including domestic firm productivity ($A_c$), the efficiency of the training process ($A_h$), and the productivity of the MNE ($A_e$). Moreover, according to Theorems \ref{dynamic_newtech} and \ref{theorem_nofixedcots}, if the host country invests in R\&D and local conditions are sufficiently favorable, it can achieve sustained long-run growth. Importantly, this outcome may arise even in the absence of FDI.

Our conclusion regarding the conditional impact of FDI on economic growth is consistent with a broad empirical literature \citep{borenz98, berth00, li05, alfaro04, alfaro10}, as explained in the introduction. For instance, using data on FDI flows to 69 developing countries over the period 1970–1989, \citet{borenz98} show that FDI contributes more to economic growth than domestic investment, particularly when the host country has reached a minimum threshold of human capital. Similarly, \citet{li05} identify a strong endogenous relationship between FDI and economic growth across 84 countries between the mid-1980s and 1999, highlighting that the growth effects of FDI increase with human capital accumulation but diminish when the technology gap between host and source countries is large. In the same vein, focusing on China between 1985 and 1996, \citet{berth00} find that provinces with higher levels of human capital benefit significantly more from FDI than less-developed regions. More recently, cross-country evidence by \citet{alfaro04} and \citet{alfaro10} confirms that the growth impact of FDI depends critically on domestic absorptive capacity, including financial development, human capital, and institutional quality.

\section{Conclusion}\label{conclusion}
We have investigated the nexus between FDI, R\&D,  and growth in a host country by using infinite-horizon optimal growth models. According to our results, a key question is not whether developing countries should attract inward FDI, but rather how they implement policies to benefit from FDI spillovers. We have proved that FDI can act as a catalyst, helping a host developing country avoid a middle-income trap and potentially achieve higher income. However, to reach sustained economic growth in the long run, the host country should focus on domestic investment and R\&D.

\appendix
\section{Formal proofs}
 \setcounter{equation}{0} 
\numberwithin{equation}{section}


\label{appendix-proof}
\begin{proof}[{\bf Proof of Lemma \ref{wage}}]
At equilibrium, the labor market clears $L^D_{e,t}=L_{e,t}$. By budget constraints \eqref{budgetconstraints}, we have $L_{e,t}<\infty$ $\forall t$. So, $L^D_{e,t}<\infty$. This implies that the wage $w_t>0$  $\forall t$ (otherwise, the profit maximization (\ref{e_MNE}) implies that $L^D_{e,t}=\infty$). Since $\alpha_h\in (0,1)$, we have Inada's condition for the function $A_hH_t^{\alpha_h}$. By combining with $w_t>0$, we must have $H_t>0$ and $L_{e,t}>0$ which in turn implies that $L^D_{e,t}>0$.  Hence, the first-order conditions (FOC) of the problem $(F_t)$ write
\begin{subequations}
\begin{align}\label{wt1}
p_n\alpha_eA_eK_{e,t}^{\alpha_e-1}(L^D_{e,t})^{1-\alpha_e}=p\\
\label{wt2}p_n(1-\alpha_e)A_eK_{e,t}^{\alpha_e}(L^D_{e,t})^{-\alpha_e}=w_t.
\end{align} \end{subequations}
By using \eqref{wt1}, we compute $K_{e,t}/L^D_{e,t}$ as a function of $p_n\alpha_eA_e,p$ and $\alpha_e$. Then we substitute it in \eqref{wt2} in order to compute the wage as shown in \eqref{wt}.
\end{proof}

\begin{proof}[{\bf Proof of Proposition \ref{s2}}] 
Point \ref{static_1} is obvious. Let us prove point \ref{static_2}. Denote $x\equiv bS^{\sigma}-\bar{x}$. Since $x>0$, there exists $ \alpha_n\in (0,1)$ such that  $bS^{\sigma}\alpha_n^{\sigma}=\bar{x}$.  Define  $K_c,N,H$ by 
$ N=(\alpha_n+\epsilon)S, pK_c=\epsilon S, H=0$, where $\epsilon >0$ such that $\alpha_n+2\epsilon=1$ (so that $N+pK_c=S$). Precisely, $\epsilon=\frac{1}{2} \Big(1-\big(\frac{\bar{x}}{bS^{\sigma}}\big)^{\frac{1}{\sigma}} \Big)$.
With such $N,K_c$, we have $bN^{\sigma}>\bar{x}$, and hence, we can verify that
\begin{align}
g(K_c,N,H)&=\Big[A_c+{a}\Big(\big(b^{\frac{1}{\sigma}}\frac{S}{2}+\frac{\bar{x}^{\frac{1}{\sigma}}}{2}\big)^{\sigma}-\bar{x}\Big)\Big]\frac{1}{p^{\alpha}}  \Big(\frac{S}{2}-\frac{\bar{x}^{\frac{1}{\sigma}}}{2b^{\frac{1}{\sigma}}}\Big)^{\alpha}
\end{align}
$g(K_c,N,H)$ is increasing in $a$ and $b$. It will be higher than $F(S)$ when $a$ and $b$ are high enough because $F(S)$ does not depend on $(a,b)$. 

Point \ref{static_3}. Under our assumption, we have $bS^{\sigma}>\lambda \bar{x}$. This implies that $\bar{x}<bS^{\sigma}/\lambda$ and $\bar{x}^{1/\sigma}<Sb^{1/\sigma}/\lambda^{1/\sigma}$. Thus, 
\begin{align}
&\Big[A_c+{a}\Big(\big(b^{\frac{1}{\sigma}}\frac{S}{2}+\frac{\bar{x}^{\frac{1}{\sigma}}}{2}\big)^{\sigma}-\bar{x}\Big)\Big]\frac{1}{p^{\alpha}}  \Big(\frac{S}{2}-\frac{\bar{x}^{\frac{1}{\sigma}}}{2b^{\frac{1}{\sigma}}}\Big)^{\alpha}\\
\geq &{a}\Big(b\frac{S^{\sigma}}{2^{\sigma}}-\frac{b}{\lambda}S^{\sigma}\Big) \frac{1}{p^{\alpha}}  \Big(\frac{S}{2}-\frac{S}{\lambda^{1/\sigma}}\Big)^{\alpha}=abS^{\alpha+\sigma}M.
\end{align}
By the definition of $F(S)$, we have $F(S)\leq A_c(\frac{S}{p_c})^{\alpha}+wA_hS^{\alpha_h}$. Our assumption $S>\bar{S}(a,b)$ ensures that $\frac{1}{2}abMS^{\alpha+\sigma}>  A_c(\frac{S}{p_c})^{\alpha}$ and $\frac{1}{2}abMS^{\alpha+\sigma}>  wA_hS^{\alpha_h}$, which imply that $abMS^{\alpha+\sigma}>F(S)$. So,  applying point \ref{static_2}, we get $N>0$. 
\end{proof}

\begin{proof}[{\bf Proof of Proposition \ref{basic}}] Since the function $G(\cdot)$ is continuous, strictly increasing, by using the standard argument in dynamic programming  \citep*{amir96},  we obtain that the optimal path $(S_t)_t$ is monotonic. 

We now prove that $S_t$ does not converge to zero. According to Lemma 3 in \cite{NguyenPham2024} (or Lemma 3.6 in \cite{kamiroy07}), we  have the Euler conditions in the form of inequality \begin{eqnarray}\label{euler1} 
 \beta u'(c_{t+1})D^-G(S_{t+1}) \geq u'(c_t)\geq \beta u'(c_{t+1})D^+G(S_{t+1}).
\end{eqnarray}
where the  Dini derivatives of function $G$ are defined by $D^+G(x)=\limsupp_{\epsilon \downarrow 0}\frac{G(x+\epsilon)-G(x)}{\epsilon}$ and $D^-G(x)=\liminff_{\epsilon \downarrow 0}\frac{G(x)-G(x-\epsilon)}{\epsilon}$.

Suppose that $\lim_{t\to \infty}S_t=0$. According to budget constraints and the fact that $G(0)=0$, we have $\lim_{t\to \infty}c_t=0$. Since $\lim_{t\rightarrow +\infty}S_t=0$, there exists $t_0$ such that $\beta D^+G(S_{t+1})>1$ for every $t\geq t_0$. Consequently, $u'(c_t)\geq u'(c_{t+1})$ and hence $c_t\leq c_{t+1}$ for every $t\geq t_0$. This leads to a contradiction to the fact that $\lim_{t\rightarrow +\infty}c_t=0$.
\end{proof}

\begin{proof}[{\bf Proof of Lemma \ref{FS}}]
(1) If $x<x^*$, then $F(x)<F(x^*)=x^*$. If $x>x^*$, then $\frac{F(x)}{x}\leq \frac{F(x^*)}{x^*}=1$ (because  $F$ is concave).

(2) It is obvious that $S_t\leq x_t $ $ \forall t$. We prove 
 $x_t\leq \bar{S}$  $\forall t$ by using the induction argument. First, we see that $x_0\leq \bar{S}.$ Second, assume that $x_s\leq \bar{S}$ $\forall s\leq t$. If $X_0\leq x^*$, then  $x_t\leq \bar{S}=x^*$, then $x_{t+1}=F(x_t)\leq F(x^*)=x^*= \bar{S}$. If $X_0> x^*$,  then $x_t\leq \bar{S}=X_0$ and hence $x_{t+1}=F(x_t)=F(x_0)\leq x_1\leq \bar{S}$.
 
\end{proof}

\begin{proof}[{\bf Proof of Theorem \ref{poverty2}}]
 We will prove, by induction argument, that $b\bar{x}_t^{\sigma}\leq \bar{x}$ and $S_t\leq x_1$ $\forall t\geq 1$. When $t=1$, we have $N_1\leq S_1\leq X_0\leq x_1$, So, $bN_1^{\sigma}\leq b\bar{S}_1^{\sigma}\leq \bar{x}$.  Assume that $b\bar{x}_t^{\sigma}\leq \bar{x}$ and $S_t\leq x_1$ $\forall t\leq T$. This implies that $N_T=0$, because otherwise we can reduce $N_T$ and augment $K_{c,T}$ in order to get a higher utility, which is a contradiction.

Since $N_T=0$, we have that $G(S_T)=F(S_T)$. Since $S_T\leq x_1$, we have $F(S_T)\leq F(x^*)=x^*$. Hence, $S_{T+1}\leq G(S_T)\leq x^*$ and therefore $b\bar{x}_{T+1}^{\sigma}\leq bS_{T+1}^{\sigma} \leq b(x^*)^{\sigma}\leq \bar{x}$. We have finished our proof.

\end{proof}

\begin{proof}[{\bf Proof of Proposition \ref{bounded_nofixedcost}}]We use the induction argument. 
Since $X_0<\bar{X}$, we have $S_1\leq c_0+S_1\leq X_0\leq \bar{X}$. So, the claim holds for $t=1$. Suppose that it holds until $t$, i.e., $S_{j}\leq \bar{X}$ $\forall j\leq t$. We have $S_{t+1}\leq G(S_t)\leq G(\bar{X})$.
By the definition of the function $G$ and condition \eqref{assumption_lowproductivity}, we have
\begin{align}
G(\bar{X})&\leq (A_c+ab\bar{X}^{\sigma})\big(\frac{\bar{X}}{p}\big)^{\alpha}+wA_hS^{\alpha_h}=\frac{A_c}{p^{\alpha}}\bar{X}^{\alpha}+\frac{ab}{p^{\alpha}}\bar{X}^{\sigma+\alpha}+wA_h\bar{X}^{\alpha_h}\\
&\leq \frac{\bar{X}}{3}+\frac{\bar{X}}{3}+\frac{\bar{X}}{3}=\bar{X}.
\end{align}
So, $S_{t+1}\leq \bar{X}$. We have finished our proof.
\end{proof}

\subsection{Proof of Theorems \ref{dynamic_newtech}}
\label{appendixA1}
The proof is quite complicated. We proceed in several steps.

\begin{lemma}\label{s4} Let Assumption \ref{assumption_basic}  be satisfied.  Assume that $\alpha+\sigma\geq 1$. For any solution $(K_c, N, H)$ of the problem $(G_S)$, denote $\theta_c\equiv \frac{pK_{c}}{S},  \theta_n\equiv \frac{N}{S}, \theta_h \equiv \frac{H}{S}.$ Then, we have 
\begin{align}\label{ratio}
\limm_{S\rightarrow \infty}\theta_c&=\frac{\alpha}{\alpha +\sigma },& \limm_{S\rightarrow \infty}\theta_n&=\frac{\sigma}{\alpha +\sigma }, &\limm_{S\rightarrow \infty}\theta_h&=0.
\end{align}
\end{lemma}


\begin{proof}[{Proof of Lemma \ref{s4}}]Let $S$ be high enough so that conditions in  point \ref{static_3} of Proposition \ref{s2} holds. In this case, we have $bN^{\sigma}-\bar{x}>0$ at optimal. It is easy to see that $\theta_c, \theta_h>0$. By consequence, we can write FOCs for the problem  $(G')$  as follows (we have FOCs even the objective function is not concave):
\begin{align}
\label{foc1}\alpha_h w A_h S^{\alpha_h}\theta_h^{\alpha_h-1}&=\lambda\\
\label{foc2}\Big(A_c+a(bS^{\sigma}\theta_n^{\sigma}-\bar{x})^+\Big)\frac{\alpha}{p^{\alpha}}\theta_c^{\alpha-1}S^{\alpha}&=\lambda\\
\label{foc3}ab\sigma S^{\sigma+\alpha}\theta_n^{\sigma-1} \big(\frac{\theta_{c}}{p}\big)^{\alpha}&=\lambda
\end{align}
where $\lambda$ is the multiplier associated to  the constraint $\theta_c+\theta_n+\theta_h\leq 1.$ Conditions \eqref{foc1} and \eqref{foc3} imply that 
\begin{align}
\frac{\alpha_h w A_hp^{\alpha}}{ab\sigma} = S^{\sigma+\alpha-\alpha_h}\theta_n^{\sigma-1} \theta_{c}^{\alpha}\theta_h^{1-\alpha_h}= (S\theta_n)^{\sigma-1} (S\theta_{c})^{\alpha}(S\theta_h)^{1-\alpha_h}
\label{stheah}
\end{align}
while \eqref{foc2} and \eqref{foc3} imply that  
$\Big(A_c+a(bS^{\sigma}\theta_n^{\sigma}-\bar{x})^+\Big)\alpha=
ab\sigma S^{\sigma}\theta_n^{\sigma-1} \theta_c.$ 
By  consequence, we obtain
\begin{align}\label{cn}
\theta_c=\frac{\alpha }{\sigma}\theta_n + \frac{\alpha \theta_n^{1-\sigma}(A_c-a\bar{x})}{ab\sigma S^{\sigma}}\text{ and }\frac{S\theta_c}{S\theta_n}=\frac{\alpha}{\sigma} + \frac{\alpha (A_c-a\bar{x})}{\sigma ab (S\theta_n)^{\sigma}}.
\end{align}
From this, we get $\lim_{S\rightarrow \infty}(\frac{\sigma\theta_c}{\alpha\theta_n}-1)\theta_n^{\sigma}=0$. This means $\lim_{S\rightarrow \infty}(\frac{\sigma\theta_c}{\alpha}-\theta_n)\theta_n^{\sigma-1}=0$. Since $\theta_n\in (0,1)$ and $\sigma\in (0,1]$, we have  $\theta_n^{\sigma-1}\geq 1$. By consequence, we obtain $\lim_{S\rightarrow \infty}(\theta_c-\frac{\alpha }{\sigma}\theta_n)=0$.


We will prove that when $S$ tends to infinity, $S\theta_h$ is bounded from above, and hence $\lim_{S\rightarrow \infty}\theta_h=0$. To do so, we firstly prove that $\liminf_{S\rightarrow \infty}\frac{(S\theta_c)^{\alpha}}{(S\theta_n)^{1-\sigma}}>0$. Indeed, according to (\ref{cn}), we have 
\begin{align}
\frac{(S\theta_c)^{\alpha}}{(S\theta_n)^{1-\sigma}}=
(S\theta_n)^{\alpha+\sigma-1} \Big(\frac{\alpha}{\sigma} + \frac{\alpha (A_c-a\bar{x})}{\sigma ab (S\theta_n)^{\sigma}}\Big)^{\alpha}
\end{align}
Suppose that there is an increasing sequence $(S_k)_k$ tending to infinity such that $\lim_{k\rightarrow \infty}\frac{(S_k\theta_c)^{\alpha}}{(S_k\theta_n)^{1-\sigma}}=0$. Notice that 
\begin{align*}
\frac{(S\theta_c)^{\alpha}}{(S\theta_n)^{1-\sigma}}=
\frac{1}{(S\theta_n)^{(1-\sigma)(1-\alpha)}} \Big(\frac{\alpha}{\sigma ab}\big[A_c + a(b(S\theta_n)^{\sigma}-\bar{x})\big]\Big)^{\alpha}\geq 
\frac{1}{(S\theta_n)^{(1-\sigma)(1-\alpha)}} \Big(\frac{\alpha}{\sigma ab}A_c\Big)^{\alpha}
\end{align*} for  $S$ high enough (because $b(S\theta_n)^{\sigma}>N$ for $S$ high enough). By combining with $\lim_{k\rightarrow \infty}\frac{(S_k\theta_c)^{\alpha}}{(S_k\theta_n)^{1-\sigma}}=0$, we have $\lim_{k\rightarrow \infty}S_k\theta_n=\infty$. However, this will follow that 
\begin{align}
\frac{(S_k\theta_c)^{\alpha}}{(S_k\theta_n)^{1-\sigma}}=
(S_k\theta_n)^{\alpha+\sigma-1} \Big(\frac{\alpha}{\sigma} + \frac{\alpha (A_c-a\bar{x})}{\sigma ab (S_k\theta_n)^{\sigma}}\Big)^{\alpha}
\end{align}
is bounded away from zero (because $\alpha+\sigma\geq1$), a contradiction. So, we get that $\liminf_{S\rightarrow \infty}\frac{(S\theta_c)^{\alpha}}{(S\theta_n)^{1-\sigma}}>0$. By combining this with 
(\ref{stheah}), we obtain that $S\theta_h$ is bounded from above and hence $\lim_{S\rightarrow \infty}\theta_h=0$.  Combining with (\ref{cn}), we obtain (\ref{ratio}).

\end{proof}

%
\begin{lemma}\label{G_F}
If $G(x)=F(x)$, then we have $D^+G(x)\geq F'(\bar{S}(a,b))>1/\beta$ for $a,b$ high enough (recall that $\bar{S}(a,b)$ is defined in point \ref{static_3} of  Proposition \ref{s2}).
\end{lemma}
\begin{proof}
If $G(x)=F(x)$, we have 
\begin{align}
D^+G(x)&=\limsupp_{\epsilon \downarrow 0}\frac{G(x+\epsilon)-G(x)}{\epsilon}=\limsupp_{\epsilon \downarrow 0}\frac{G(x+\epsilon)-F(x)}{\epsilon}\\
&\geq \limsupp_{\epsilon \downarrow 0}\frac{F(x+\epsilon)-F(x)}{\epsilon}=F'(x).
\end{align}
According to point \ref{static_3} of  Proposition \ref{s2}, $G(x)=F(x)$ implies that $x<\bar{S}(a,b)$. Since $F'$ is decreasing, we have $F'(x)>F'(\bar{S}(a,b))>1/\beta$ for $a,b$ high enough (because $F'(0)=\infty$; see Lemma \ref{function_F}). By consequence, $D^+G(x)>1/\beta$ for   $a,b$ high enough.
    
\end{proof}


\begin{lemma}\label{DGS}
Assume that $a\bar{x}-A_c\geq 0$.  Consider a point $S>0$ satisfying  $G(S)>F(S)$. We have 
$D^+G(S)\geq \Gamma(a,b,\bar{x}) \equiv a^{1-\alpha}b^{\frac{1-\alpha}{\sigma}} \sigma \bar{x}^{\frac{\sigma+\alpha-1}{\sigma}}\Big(\frac{\alpha A_c}{p\sigma \bar{x}}\Big)^{\alpha}.$ 
By consequence, when $a$ and $b$ are high enough, we have $\beta D^+G(S)>1$ $\forall S>0$  satisfying  $G(S)>F(S)$.
\end{lemma}

\begin{proof}[Proof of Lemma \ref{DGS}]

Our proof consists of two steps.\\
{\bf Step 1.}  Let $(K_c,N,H)$ be such that $G(S)= (A_c+a(bN^{\sigma}-\bar{x})^+)K_{c}^{\alpha}+wA_hH^{\alpha_h}$. Since $G(S)>F(S)$, we have $bN^{\sigma}> \bar{x}$ at optimal. It is obvious that $K_c>0$ and $H>0$ in optimal. So, this solution is interior. By consequence, we can use the Lagrangian method. Let $\lambda$ be the multiplier associated to the constraint $pK_c+N+H\leq S$, we have FOCs
\begin{align}\label{FOC_GS}
(abN^{\sigma}-(a\bar{x}-A_c)) \alpha K_c^{\alpha-1}&=p \lambda, &
ab\sigma N^{\sigma-1} K_c^{\alpha}&=\lambda,&
\alpha_h w A_h H^{\alpha_h-1}&=\lambda.
\end{align}

FOCs imply that $\alpha (abN^{\sigma}-(a\bar{x}-A_c)) =p ab\sigma N^{\sigma-1} K_c$ and hence
\begin{align}
\frac{\alpha}{\sigma}N\geq pK_c&=\frac{\alpha}{\sigma}N \Big(1-\frac{a\bar{x}-A_c}{abN^{\sigma}} \Big)>N\frac{\alpha A_c}{\sigma a\bar{x}}
\end{align}
because $a\bar{x}-A_c\geq 0$ and $bN^{\sigma}\geq \bar{x}$. By consequence, we have
\begin{align}\label{K_cN}
\frac{K_c}{N}> \frac{\alpha A_c}{p\sigma a\bar{x}}.
\end{align}

{\bf Step 2.}  We rewrite $G(S)= (abN^{\sigma}-d)K_{c}^{\alpha}+wA_hH^{\alpha_h}$, where  $d\equiv a\bar{x}-A_c\geq 0$.

Denote $\epsilon_1\equiv \epsilon/(p+2).$ Since $pK_c+N+H=S$, we have $p(K_c+\epsilon_1)+(N+\epsilon_1)+(H+\epsilon_1)=S+\epsilon$. Then, by the definition of $G(S+\epsilon)$, we have
\begin{align*}
G(S+\epsilon)&\geq (ab(N+\epsilon_1)^{\sigma}-d)(K_{c}+\epsilon_1)^{\alpha}+wA_h(H+\epsilon_1)^{\alpha_h}.
\end{align*}
This implies that $G(S+\epsilon)-G(S)
\geq (ab(N+\epsilon_1)^{\sigma}-d)(K_{c}+\epsilon_1)^{\alpha}+wA_h(H+\epsilon_1)^{\alpha_h}-(abN^{\sigma}-d)K_{c}^{\alpha}-wA_hH^{\alpha_h}.$ 
By dividing both sides by $\epsilon$ and using  $\epsilon=\epsilon_1(p+2)$, we have
\begin{align*}
\frac{G(S+\epsilon)-G(S)}{\epsilon}
\geq &\frac{ab}{p+2}\frac{(N+\epsilon_1)^{\sigma}-N^{\sigma}}{\epsilon_1}(K_{c}+\epsilon_1)^{\alpha}\\
&+\frac{abN^{\sigma}-d}{p+2}\frac{(K_{c}+\epsilon_1)^{\alpha}-K_{c}^{\alpha}}{\epsilon_1}+\frac{wA_h}{p+2}\frac{(H+\epsilon_1)^{\alpha_h}-H^{\alpha_h}}{\epsilon_1}.
\end{align*}
Let $\epsilon$ decrease to $0$ and by using \eqref{FOC_GS}, we obtain
\begin{align}
D^+G(S)\geq ab\sigma N^{\sigma-1}K_{c}^{\alpha}\Big(\frac{1}{p+2}+\frac{p}{p+2}+\frac{1}{p+2}\Big)=ab\sigma N^{\sigma-1}K_{c}^{\alpha}.
\end{align}
By combining this with $bN^{\sigma}> \bar{x}$ and condition \eqref{K_cN}, we get that
\begin{align}
D^+G(S)&\geq \sigma abN^{\sigma-1}K_{c}^{\alpha}= \sigma abN^{\sigma+\alpha-1}\big(\frac{K_{c}}{N}\big)^{\alpha} >\sigma ab\big(\frac{\bar{x}}{b}\big)^{\frac{\sigma+\alpha-1}{\sigma}}\Big(\frac{\alpha A_c}{p\sigma a \bar{x}}\Big)^{\alpha}\\
&=a^{1-\alpha}b^{\frac{1-\alpha}{\sigma}} \sigma \bar{x}^{\frac{\sigma+\alpha-1}{\sigma}}\Big(\frac{\alpha A_c}{p\sigma \bar{x}}\Big)^{\alpha}.
\end{align}
Thus, it is easy to see that $D^+G(S)>1/\beta$ when $a,b$ are high enough.

\end{proof}

We are now ready to prove Theorem \ref{dynamic_newtech}. When $a$ is high enough, we have $a\bar{x}-A_c>0$. According to Lemmas \ref{G_F} and \ref{DGS}, we have $\beta D^+G(S)\geq \beta \min\{F'(\bar{S}(a,b)), \Gamma (a,b,\bar{x})\}$ $\forall S>0$. So, $\beta D^+G(S)>1$ when $a$ and $b$ are high enough.  According to Proposition 4.6 in  \citet*{kamiroy07}, we get that $\lim_{t\rightarrow \infty}S_t=\infty$. According to Lemma \ref{s4}, we obtain \eqref{converge_share} in Theorem \ref{dynamic_newtech}.

\subsection{Proof of Theorem \ref{theorem_nofixedcots} and an additional result}
\label{proofs_nofixedcost}
We need intermediate steps.

\begin{lemma}\label{function_G1} Consider the function $G_1: \rr^+\to\rr^+$ defined by $G_1(x)\equiv \max\{(A_c+ab N^{\sigma})K_c^{\alpha}: pK_c+N\leq x, N\geq 0, K_c\geq 0\}$. Let $x>0$ and $(K_c,N)$ be a solution to this maximization problem.

\begin{enumerate}
\item If $\sigma\geq 1$ and $\sigma ab x^{\sigma}\leq \alpha A_c$, then $N=0$.
\item If $\sigma\in (0,1)$, then $N>0$ and it is a unique solution to $(\alpha+\sigma)ab N+\alpha A_cN^{1-\sigma}=\sigma ab x.$
\end{enumerate}
\end{lemma}

\begin{proof}

We also observe that $(K_c,N)$ is a solution to the following problem: $\max\{\ln (A_c+abN^{\sigma})+\alpha \ln(K_c): pK_c+N\leq x, (K_c,N)\in \rr^2_+\}$.

\paragraph{Case 1: $\sigma\geq 1$ and $\sigma ab x^{\sigma}\leq \alpha A_c$.} If $N>0$, then we have the FOC
\begin{align}\label{G1_N_equation}
\frac{\sigma ab N^{\sigma-1}}{A_c+ab N^{\sigma}}=\frac{\alpha}{x-N}\iff (\alpha+\sigma)abN+\alpha A_cN^{1-\sigma}=\sigma ab x.
\end{align}

Since $\sigma\geq 1$ and $N\leq x$, we have $N^{1-\sigma}\geq x^{1-\sigma}$. Combining with $N>0$, we obtain $\sigma ab x> \alpha A_cS^{1-\sigma}$, or, equivalently, $\sigma ab x^{\sigma}>\alpha A_c$, a contradiction. So, we conclude that $N=0$.

\paragraph{Case 2: $\sigma\in (0,1)$.}   We observe that $\lim_{N\to 0}\frac{\partial \ln (A_c+abN^{\sigma})}{\partial N}=\infty$. By consequence, $N>0$ in optimal. Then, by the same argument,  $N$ is the solution to the equation \eqref{G1_N_equation}.

\end{proof}

\begin{lemma}\label{nofixedcots_intermediate1}
Assume that $\bar{x}=0$. 
Let $(K_c,N,H)$ be a solution to the problem $(G_S)$ described in \eqref{problemGS}.
\begin{enumerate}
\item Let $\sigma\geq 1$. If $\sigma abS^{\sigma}\leq \alpha A_c$, then $N=0$ and $G(S)=F(S)$.
\item Let $\sigma<1$. In optimal, we have $N>0$.
\end{enumerate} 
\end{lemma}
\begin{proof}
Let $(K_c,N,H)$ be a solution. Then, we have $K_c,H>0$ thanks to Inada's condition. Then $(K_c,N$ is a solution to the problem $\max\{(A_c+ab N^{\sigma})K_c^{\alpha}: pK_c+N\leq S-H, N\geq 0, K_c\geq 0\}$. Then applying Lemma \ref{function_G1}, we obtain our results.

\end{proof}

\begin{lemma}\label{DGS}
Assume that $\bar{x}=0$, $\sigma\in (0,1)$ and \eqref{highproductivity_nofixedcost} holds. Then, we have $D^+G(S)>1/\beta$ $\forall S>0$. By consequence, when $a$ and $b$ are high enough, we have $\beta D^+G(S)>1$ $\forall S>0$.
\end{lemma}

\begin{proof}

Let $(K_c,N,H)$ be such that $G(S)= (A_c+abN^{\sigma})K_{c}^{\alpha}+wA_hH^{\alpha_h}$. It is obvious that $K_c>0$ and $H>0$ in optimal. According to Lemma \ref{nofixedcots_intermediate1}, we have $N>0$. So, this solution $(K_c,N,H)$ is interior. By consequence, we can use the Lagrangian method. Let $\lambda$ be the multiplier associated to the constraint $pK_c+N+H\leq S$, we have FOCs
\begin{align}\label{FOC_GS}
(A_c+abN^{\sigma}) \alpha K_c^{\alpha-1}&=p \lambda, &
ab\sigma N^{\sigma-1} K_c^{\alpha}&=\lambda,&
\alpha_h w A_h H^{\alpha_h-1}&=\lambda.
\end{align}
The two first conditions imply
\begin{align} \label{KcN}
\alpha (A_c+abN^{\sigma}) =p ab\sigma N^{\sigma-1} K_c, \quad 
K_c&\geq N\frac{\alpha }{p\sigma}.
\end{align}

By using the same argument in the second step of the proof of Lemma \ref{DGS}, we have 
\begin{align}\label{boundDG}
D^+G(S)&\geq \sigma abN^{\sigma-1}K_{c}^{\alpha} \text{ and } D^+G(S)\geq \frac{1}{p}(A_c+abN^{\sigma})\alpha K_{c}^{\alpha-1}.
\end{align}


We consider two cases.
\begin{enumerate}
    \item If $N\leq 1$, then \eqref{boundDG} and \eqref{KcN} give that
\begin{subequations}
    \begin{align}
D^+G(S)&\geq \frac{\alpha}{p}\frac{(A_c+abN^{\sigma})}{K_{c}^{1-\alpha}}=\frac{\alpha}{p}\frac{(A_c+abN^{\sigma})}{\big(\frac{\alpha }{p\sigma ab }(A_c+abN^{\sigma})N^{1-\sigma}\big)^{1-\alpha}}\\
&=\big(\frac{\alpha}{p}\big)^{\alpha}(\sigma ab)^{1-\alpha}\frac{(A_c+abN^{\sigma})^{\alpha}}{N^{(1-\alpha)(1-\sigma)}}\geq \big(\frac{\alpha}{p}\big)^{\alpha}(\sigma ab)^{1-\alpha}A_c^{\alpha}>1/\beta,
\end{align}
\end{subequations}
where the last inequality is actually our assumption \eqref{highproductivity_nofixedcost}.
\item If $N>1$, by using \eqref{boundDG}, \eqref{KcN}, and our assumption \eqref{highproductivity_nofixedcost}, we have
\begin{subequations}
    \begin{align}
D^+G(S)&\geq \sigma abN^{\sigma-1}K_{c}^{\alpha}= \sigma abN^{\alpha+\sigma-1}\big(\frac{K_{c}}{N}\big)^{\alpha}\\
&\geq \sigma abN^{\alpha+\sigma-1}\big(\frac{\alpha}{p\sigma}\big)^{\alpha}\geq \sigma ab\big(\frac{\alpha}{p\sigma}\big)^{\alpha}=\sigma^{1-\alpha} ab\big(\frac{\alpha}{p}\big)^{\alpha}>1/\beta.
\end{align}
\end{subequations}
\end{enumerate}
\end{proof}

\begin{proof}[{\bf Proof of Theorem \ref{theorem_nofixedcots}}]
Observe that Lemma \ref{s4} remains valid when there is no fixed costs ($\bar{x}=0$). Notice also that points  \ref{static_2} and \ref{static_3} of Proposition \ref{s2} still hold.  

By Lemma \ref{DGS}, we have $\beta D^+G(S)>1$ $\forall S>0$. So, applying  Proposition 4.6 in  \citet*{kamiroy07}, we get that $\lim_{t\rightarrow \infty}S_t=\infty$. Then, according to Lemma \ref{s4}, we obtain \eqref{converge_share} in Theorem \ref{dynamic_newtech}.  

Observe that, since $\sigma\in (0,1)$, by Lemma \ref{nofixedcots_intermediate1}, we have $N_t>0$ $\forall t$.
\end{proof}

\begin{proposition}[additional result]
\label{theorem_nofixedcots_1}Let the conditions in Assumption \ref{assumption_basic} hold, except that the restriction $\sigma \in (0,1]$ is replaced by $\sigma \geq 1$. Assume that there is no fixed cost ($\bar{x}=0$). If $X_0\leq x^*$ and $\sigma ab(x^*)^{\sigma}\leq \alpha A_c$. Then, the conclusions in Theorem \ref{poverty2} holds. Precisely,  in optimal, $N_t=0$ $\forall t$ and $\lim_{t\rightarrow \infty}S_t=S_b$. 
\end{proposition} 
\begin{proof}[{\bf Proof of Proposition \ref{theorem_nofixedcots_1}}]   
 We have $S_1\leq X_0\leq x^*$. This implies that $\sigma abS_1^{\sigma}\leq \alpha A_c$. By Lemma \ref{nofixedcots_intermediate1}, we have $N_1=0$ and $G(S_1)=F(S_1)$. 
Then, $S_2\leq G(S_1)=F(S_1)\leq F(x^*)=x^*$. By induction, we have $N_t=0$, $G(S_1)=F(S_t)$ and $S_t\leq x^*$ $\forall t\geq 1$. This implies that $\lim_{t\to \infty}S_t=S_b$.
\end{proof}

\subsection{Proof of Theorem \ref{drs}}
\label{proof_drs}
Recall that $ c_{t}+S_{t+1}\leq G(S_t)$ $\forall t\geq 1$ and $c_0+S_1\leq X_0$. So, $S_{t+1}\leq G(S_t)$ $\forall t\geq 1$ and $S_1\leq X_0$.
By the definition of the function $G$ in \eqref{problemGS}, we observe that 
\begin{align*}
G(S)&\leq (A_c+abS^{\sigma})\frac{1}{p^{\alpha}} S^{\alpha}+wA_hS^{\alpha_h}
\leq \begin{cases}
 \frac{A_c+ab}{p^{\alpha}}+wA_h \text{ if } S\leq 1\\
 \Big(\frac{A_c+ab}{p^{\alpha}}+wA_h\Big) S^{\max\{\alpha+\sigma, \alpha_h\}}  \text{ if } S\geq 1.
\end{cases}
\end{align*}
By using $\max(\alpha+\sigma, \alpha_h)<1$, we can prove that $S_t$ is bounded from above (see, for instance, Lemma 1 in \cite{LeVanPham2016}). Indeed, denote $A\equiv {(A_c+ab)}/{p^{\alpha}}+wA_h$ and $\gamma\equiv \max\{\alpha+\sigma, \alpha_h\}$.  Since $\gamma\in (0,1)$, we can take $\bar{X}$ be such that  $\bar{X}>\max\{1,A,X_0\}$ and $A\bar{X}^{\gamma-1}<1$.  We claim that $S_t\leq \bar{X}$ $\forall t$. 

We have $S_1\leq \bar{X}$. If $S_1\leq 1$, then $S_2\leq G(S_1)\leq A\leq \bar{X}$.  If $1<S_{1}\leq \bar{X}$, then $S_2\leq G(S_1)\leq AS_1^{\gamma}\leq A \bar{X}^{\gamma}\leq\bar{X}$. 
Combining two cases, we have $S_2\leq \bar{X}$.
 By induction, we obtain $S_t\leq \bar{X}$ $\forall t$.

By Proposition \ref{basic}, the sequence $S_t$ is monotonic. By consequence, it must converge to a finite value. 



{\small
\bibliographystyle{apalike}

\bibliography{refs}
}

\end{document}